\documentclass[10pt]{article}
\AtBeginDocument{%
  }

\usepackage{amsthm}

\usepackage{amssymb,stmaryrd}
\usepackage{graphics}
\usepackage{amsfonts}
\usepackage{bbm}
\usepackage{tikz}
\usepackage{color}
\usepackage{amsmath}

\newcommand{\vx} { {\bf x}}
\newcommand{\vy} { {\bf y}}
\newcommand{\vz} { {\bf z}}

\newcommand{\one} { {\bf 1}}
\newcommand{\hx} { {\hat{\bf x}}}

\def\e{{\boldsymbol \varepsilon}}

\def\<#1>{\langle#1\rangle}

\def\F{\mathbb F}
\def\Z{\mathbb Z}
\def\N{\mathbb N}
\def\Q{\mathbb Q}

\def\a{\alpha}
\def\w{\omega}
\def\si{\sigma}
\def\lam{\lambda}
\def\la{\langle}
\def\ra{\rangle}
\def\tta{{\tilde\theta}}
\def\e{\varepsilon}

\newcommand{\bN} { {\mathbb{N}}}
\newcommand{\bC} { {\mathbb{C}}}
\newcommand{\bQ} { {\mathbb{Q}}}
\newcommand{\bZ} { {\mathbb{Z}}}

\newcommand{\bF} { {\mathbb{F}}}
\newcommand{\bK} { {\mathbb{K}}}
\newcommand{\bE} { {\mathbb{E}}}

\DeclareMathOperator{\rank}{rank}

\newtheorem{thm}{Theorem}[section]

\newtheorem{lem}[thm]{Lemma}
\newtheorem{prop}[thm]{Proposition}

\newtheorem{rem}[thm]{Remark}
\newtheorem{exam}[thm]{Example}

\newtheorem{problem}[thm]{Problem}

\makeatletter \@addtoreset{equation}{section}
\begin{document}

\title{On the Summability Problem of Multivariate Rational Functions in the Mixed Case
\thanks{This paper is dedicated to Professor George Labahn on the occasion of his 75th birthday.
	S.\ Chen, H.\ Fang, and Y.\ Wang were partially supported by the National Key R\&D Program of China (No.\ 2023YFA1009401), the NSFC grant (No.\ 12271511), the CAS Project for Young Scientists in Basic Research (No.\ YSBR-034) and the Strategic Priority Research Program of the Chinese Academy of Sciences (No.\ XDB0510201). L.\ Du was supported by the Austrian FWF grant 10.55776/PAT9952223. Y.\ Wang was supported by the Austrian FWF grant 10.55776/I6130.
	This work was supported by the International Partnership Program of Chinese Academy of Sciences (No.\ 167GJHZ2023001FN)
}
}


\author{
	Shaoshi Chen$^{a, b}$, Lixin Du$^{c}$, \\
	\bigskip
	Hanqian Fang$^{a, b}$ and Yisen Wang$^{a,b,d}$\\
	$^a$KLMM,\, Academy of Mathematics and Systems Science, \\ Chinese Academy of Sciences, \\100190 Beijing, China\\
	$^b$School of Mathematical Sciences, \\University of Chinese Academy of Sciences,\\100049 Beijing, China\\
	$^c$Institute for Algebra, \\
	Johannes Kepler University,\\
	4040 Linz, Austria\\
	$^d$RICAM, Austrian Academy of Sciences, \\
	\medskip
	4040 Linz, Austria\\	
	{\sf schen@amss.ac.cn, lixindumath@gmail.com}\\
	{\sf hqfang\_math@163.com, wangyisen@amss.ac.cn}
}

\maketitle

\begin{abstract}
  Continuing previous work, this paper focuses on the summability problem of multivariate rational functions in the mixed case in which both shift and $q$-shift operators can appear. Our summability criteria rely on three ingredients including orbital decompositions, Sato's isotropy groups, and difference transformations. This work settles the rational case of the long-term project aimed at developing algorithms for symbolic summation of multivariate functions.   
\end{abstract}

\maketitle

\section{Introduction}

As a classical and active topic in symbolic computation, symbolic summation aims at providing algorithmic tools for verifying or discovering  identities and closed forms for various sums from combinatorics~\cite{PWZbook1996, KauersPaule2011, Koepf2014}, computer science~\cite{GKP1994} and theoretical physics~\cite{Schneider2006, Schneider2012}. The summability problem determines whether a given function (or sequence) is a difference of another function (or sequence) so that the discrete Newton--Leibniz formula can be applied to compute definite sums.  
Its continuous analogue in the differential setting was related to the effective computation of de Rham cohomologies. For instance, Picard and Simart in their book~\cite[pp.\ 475--479]{Picard1897} proposed to decide whether a rational function $f(x, y, z)\in \bC(x, y, z)$ of three variables can be written as
\[f(x, y, z) = \frac{\partial u}{\partial x}+ \frac{\partial v}{\partial y}  + \frac{\partial w}{\partial z},   \]
where $u, v, w \in \bC(x, y, z)$. This is still an open problem and the related results can be found in~\cite{Griffiths1969,Dimca1991,Lairez2016,ChenCoxWang2026}. The goal of this paper is to study the discrete and $q$-discrete analogues of Picard’s problem for multivariate rational functions.

The development of algorithms for symbolic summation dates back to the early 1970s with significant advances in subsequent decades~\cite[Chapter~23]{MCA2003}. 
In the univariate case, the summability problem was studied by Abramov~\cite{Abramov1971, Abramov1975, Abramov1995b} for rational functions and also by Gosper~\cite{Gosper1978} for hypergeometric terms. 
Karr extended Risch's algorithm to the setting of so-called $\Pi\Sigma$-extensions~\cite{Karr1981, Karr1985} with a series of further developments by Schneider and his collaborators~\cite{Schneider2004, Schneider2016} motivated by computational problems in quantum field theory~\cite{Schneider2012, Schneider2013QFT}.
Extending these results to the multivariate setting was proposed as an intriguing problem by Andrews and Paule~\cite{Andrews1993}, since its solution would help us reduce multiple sums to single sums. Initial progress was made in~\cite{ChenHouMu2006} which provides some necessary conditions on the summability of bivariate hypergeometric terms. The summability problem in the case of bivariate rational functions was solved in~\cite{ChenSinger2014} with later algorithmic improvements in~\cite{HouWang2015,Wang2021}. Beyond the bivariate case, this problem has been studied for binomial sums~\cite{BostanLairezSalvy2017}, and a complete solution to the summability problem of multivariate rational functions involving {ordinary} shift operators was recently given in~\cite{ChenDuFang2025}.

As a direct and complete generalization of the results in~\cite{ChenDuFang2025}, this paper studies the summability problem of multivariate rational functions in the mixed case in which both shift and $q$-shift operators can appear.Although the overall structure and strategy remain the same, several technical modifications are required. The first new idea is to introduce the concepts of normal and special polynomials from symbolic integration, following~\cite{ChenDuGaoHuangLi2025}. In the ordinary shift case, every irreducible polynomial is normal, while in the $q$-shift case, a special polynomial may occur. Therefore, by treating this special case separately, we can handle both types of operators within a unified framework. Another new ingredient is a direct decomposition of isotropy groups into an additive difference part and a renormalized $q$-difference part, which allows one to compute a basis of an isotropy group componentwise.
Furthermore, there is a conceptual novelty in the final step, namely the construction of difference transformations. In the ordinary shift case, the relevant transformation is affine-linear, while in the $q$-shift case, a multiplicative change of variables is needed; this may involve passing to an algebraic closure before being pulled back. Finally, we construct a single mixed endomorphism that combines the classical additive transformation with the new multiplicative one.

The remainder of this paper is organized as follows. In Section~\ref{SEC:pre}, we introduce the basic definitions and notation to formally state the rational summability problem. Section~\ref{SEC:adddecomp} shows that the problem can be reduced to that of simple fractions via orbital decompositions. Section~\ref{SEC:iso} presents Sato's theory of isotropy groups, {which is then used in Section~\ref{SEC:criteria} to derive the summability criteria, thereby reducing the problem to a summability problem with fewer operators in a more general setting}. In Section~\ref{SEC:tranformation}, we construct an $\bF$-endomorphism to transform the general ($q$-)summability problem into a standard form. Section~\ref{SEC:examples} presents two examples that illustrate {potential} applications in verifying the convergence and irrationality of given series. We conclude the paper in Section~\ref{SEC:conclusion} by proposing some problems for future studies.

\section{Preliminaries}\label{SEC:pre}
The goal of this section is to introduce some notation and state the main problem addressed in this paper. 
Let $\bF$ be a field of characteristic zero and $\bF(\vx)$ be the field of rational functions in $\vx=\{x_1,\ldots,x_n\}$ over~$\bF$. We use $\hx_1$ to denote the $n-1$ variables~$x_2,\ldots,x_n$. For each~$v\in \vx$, the shift operator $\si_{v}$ is the $\bF$-automorphism of $\bF(\vx)$ defined by $\si_{v}(v)=v+1$ and $\si_{v}(w)=w$ for all~$w\in\vx\setminus\{v\}$. Let $q\in\bF^*:=\bF\setminus\{0\}$ be such that $q^m\neq1$ for all nonzero~$m\in\Z$. The $q$-shift operator $\tau_{q,v}$ is defined as the $\bF$-automorphism of $\bF(\vx)$ such that $\tau_{q,v}(v)=q\cdot v$ and $\tau_{q,v}(w)=w$ for all~$w\in\vx\setminus\{v\}$. Let $\theta_v\in\{\si_v,\tau_{q,v}\}$ and $G=\langle \theta_{x_1}, \ldots, \theta_{x_n}\rangle$ be the free multiplicative abelian group generated by the operators $\theta_{x_1},\ldots,\theta_{x_n}$. The group algebra $\bF[G]$ over the field $\bF$ consists of all finite linear combinations $\sum a_\theta \theta$ with $a_{\theta}\in \bF$ and~$\theta\in G$. For each~$\theta\in G$, we use $\tta$ to replace $c\cdot\theta$ for some $c\in\bF^*$ for short and the difference operator $\tta-{\bf 1}$ is denoted by~$\Delta_{\tta}$, where ${\bf 1}$ stands for the identity map on~$\bF(\vx)$. 
The main task of this paper is to solve the following problem.
\begin{problem}\label{PROB:summabilityproblem}
	Given a rational function $f\in \bF \left(\mathbf{x}\right)$ and $\tta_{x_i}=c_i\cdot\theta_{x_i}$ for some constant $c_i \in \bF^*$ with $i=1,\ldots,n$, decide whether there exist $g_1,\ldots,g_n\in\bF(\vx)$ such that
	\begin{equation}\label{EQ:summbility}
		f=\Delta_{\tta_{x_1}}(g_1)+\cdots+\Delta_{\tta_{x_n}}(g_n).
	\end{equation}
	If such $g_i$'s exist, we say that $f$ is {\em ($\tta_{x_1},\ldots,\tta_{x_n}$)-summable} in $\bF(\vx)$.
\end{problem}
Let $\vx=\vy\cup\vz$ with $\vy=\{y_1,\ldots,y_k\}, \vz=\{z_1,\ldots,z_m\}$ and $n=k+m$. If $\theta_{y_j}=\si_{y_j}$ for $1\leq j\leq k$,  $\theta_{z_\ell}=\tau_{q,z_\ell}$ for $1\leq \ell\leq m$ and $c_i=1$ for $1\leq i\leq n$, then Problem~\ref{PROB:summabilityproblem} is reduced to deciding whether $f$ is  $(\si_{y_1},\ldots,\si_{y_k},\tau_{q,z_1},\ldots,\tau_{q,z_m})$-summable.
In general, let $\theta_1,\ldots,\theta_r$ be a family of independent elements in $G$, that means if $\theta_1^{\ell_1}\cdots\theta_r^{\ell_r}=\bf1$ for some $\ell_1,\ldots,\ell_r\in \Z$ then $\ell_i=0$ for all~$1\leq i\leq r$. Let $\tta_i=c_i\theta_i$ for some $c_i\in\bF^*$ with $i=1,\ldots,r$.  A rational function $f\in\bF(\vx)$ is called {\em $(\tta_1,\ldots,\tta_r)$-summable} if $f=\Delta_{\tta_1}(g_1)+\cdots+\Delta_{\tta_r}(g_r)$ for some $g_1,\ldots,g_r\in \bF(\vx)$. The $(\tta_1,\ldots,\tta_r)$-summability problem is to decide whether $f$ is $(\tta_1,\ldots,\tta_r)$-summable.

Let $\bE=\bF(\hx_1)$. As an analog in the differential case (see~\cite[Definition~1.3]{BRONSTEIN1990}), we say that $p\in\bE[x_1]$ is {\em normal} with respect to $\theta_{x_1}$ if $\gcd(p,\theta_{x_1}^\ell(p))=1$ for any $\ell\in\Z\setminus\{0\}$. A non-normal polynomial with respect to $\theta_{x_1}$ is called {\em special} with respect to $\theta_{x_1}$. Normal polynomials and special polynomials will be further discussed in Sections \ref{SEC:adddecomp} and~\ref{SEC:criteria}. 

Let $\theta_{\vx}=(\theta_{x_1},
\ldots,\theta_{x_n})$ and $\alpha=(a_1,\ldots,a_n)\in\Z^n$. We introduce a short notation { $\theta_{\vx}^{\alpha}:=\theta_{x_1}^{a_1}\cdots\theta_{x_n}^{a_n}$} to denote an element in $ G$. For $\mathbf{c}=(c_1,\ldots,c_n)\in\bF^n$, define $\mathbf{c}^\alpha=c_1^{a_1}\cdots c_n^{a_n}$ to be a constant in $\bF$.
\section{Additive decompositions}\label{SEC:adddecomp}
In this section, we shall introduce a direct sum decomposition of $\bF(\vx)$ into $G$-invariant subspaces, and thus reduce the summability problem of general rational functions to the case of simple fractions.

\subsection{Orbital decompositions}
Let $p,q\in \bF[\vx]$ be two irreducible polynomials in $x_1$ over $\bE=\bF(\hx_1)$. We say that $p,q$ are {\em associate} if $p=c\cdot q$ for some~$c\in \bF^*$. This is an equivalence relation. Let $\Omega$ be the set of all such equivalence classes $[p]$ and $K$ be a subgroup of $G$. First of all, we shall introduce a group action of $K$ on the set $\Omega$. Since each $\theta\in K$ is an $\bF$-automorphism of $\bF(\vx)$, this naturally defines a group action on $\Omega$ by mapping $[p]$ into $[\theta(p)]$. We call the set
\[[p]_K:= \left\{[\theta(p)]\mid \theta\in K\right\}\]
the  {\em $K$-orbit} of $p$. The polynomials $p,q$ are said to be {\em $K$-equivalent} if $[p]_K=[q]_K$, denoted by $p\sim_K q$. Then $p,q$ are $K$-equivalent if and only if $p=c\cdot \theta(q)$ for some $c\in\bF^*$ and $\theta\in K$. The relation $\sim_K$ is an equivalence relation.

Now we shall decompose $\bF(\vx)$ as a vector space over $\bE$. Let $K=G$ and by the above group action, the set $\Omega$ can be partitioned into the disjoint union of the distinct $G$-orbits. Given an irreducible polynomial $d\in\bF[\vx]$ with $\deg_{x_1}(d)>0$ and a positive integer~$j$, we define an $\bE$-subspace $V_{[d]_G,j}$ of $\bF(\vx)$
spanned by all of the fractions $a/b^j$ with $a\in \bE[x_1]$, $b\sim_G d$, and~$\deg_{x_1}(a)<\deg_{x_1}(d)$.
For any fraction in $V_{[d]_G,j}$, the irreducible factors of its denominator are in the same $G$-orbit as $d$. By the irreducible partial fraction decomposition, any rational function $f\in\bF(\vx)$ can be uniquely decomposed as
$f=f_0+f_1+\cdots+f_s $
with $f_0\in\bE[x_1]$ and $f_1,\ldots,f_s$ in distinct $V_{[d]_G,j}$ spaces. Therefore, $\bF(\vx)$ admits the following direct sum decomposition:
\begin{equation}\label{EQ:adddecomp}
	\bF(\vx)=\bE[x_1]\bigoplus\Big(\bigoplus_{j\in\N^+}\bigoplus_{[d]_G\in T}V_{[d]_G,j}\Big),
\end{equation}
where $\N^+:=\N\setminus\{0\}$ and $T$ is the set of all distinct $G$-orbits such that $\Omega$ is the disjoint union of $T$. We use $f_0$ and $f_{[d]_G,j}$ to denote the components of $f$ in $\bE[x_1]$ and $V_{[d]_G,j}$, respectively.
Such a direct sum decomposition is called the \emph{orbital decomposition} of~$\bF(\vx)$ with respect to the variable~$x_1$ and the group~$G$.
\begin{lem}\label{LEM:G-invariant}
	If $f\in V_{[d]_G,j}$ and $P\in \bE[G]$, then $P(f)\in V_{[d]_G,j}$.
\end{lem}
\begin{proof}
	Let $f=\sum a_i/b_i^j$ with $b_i\sim _G d$, $\deg_{x_1}(a_i)<\deg_{x_1}(d)$ and $P=\sum{p_\theta}\theta$ with $p_\theta\in\bE$. For any $\theta\in G$, we have $\theta(b_i)$ and $d$ are still in the same $G$-orbit and $\deg_{x_1}(\theta(a_i))<\deg_{x_1}(d)$. Then $\frac{p_\theta \theta(a_i)}{\theta(b_i)^j}$ is in $V_{[d]_G,j}$. So $P(f)$ lies in~$ V_{[d]_G,j}$ by the linearity.
\end{proof}
\begin{lem}\label{LEM:red1}
	Let $f\in\bF(\vx)$. Then $f$ is $(\tta_{x_1},\ldots,\tta_{x_n})$-summable if and only if $f_0$ is $(\tta_{x_1},\ldots,\tta_{x_n})$-summable in~$\bE[x_1]$ and $f_{[d]_G,j}$ is $(\tta_{x_1},\ldots,\tta_{x_n})$-summable for all $[d]_G\in T$ and $j\in \N^+$.
\end{lem}
\begin{proof}
	The sufficiency is due to the additivity of ($q$-)shift operators. For the necessity, suppose $f=\sum_{i=1}^n\Delta_{\tta_{x_i}}(g^{(i)})$ with $g^{(i)}\in\bF(\vx)$. By the additive decomposition of rational functions in~\eqref{EQ:adddecomp}, we can write $f, g^{(i)}$ in the form
	\small
	\begin{equation*}
		f=f_0+\sum_{j}\sum_{[d]_G} f_{{[d]_G,j}} \text{ and } g^{(i)}=g^{(i)}_{0}+\sum_{j}\sum_{[d]_G} g^{(i)}_{{[d]_G,j}} \text{ for } 1\leq i\leq n.
	\end{equation*}
	\normalsize
	By the additivity of $\Delta_{\tta_{x_i}}$, we see that
	\[f=\sum_{i=1}^n\Delta_{\tta_{x_i}} \big(g^{(i)}_{0}\big)+\sum_{j}\sum_{[d]_G}\bigg(\sum_{i=1}^n\Delta_{\tta_{x_i}}\big(g^{(i)}_{[d]_G,j}\big)\bigg).\]
	From Lemma~\ref{LEM:G-invariant}, it is another expression of $f$ with respect to $V_{[d]_G,j}$. By the uniqueness of the orbital decomposition~\eqref{EQ:adddecomp}, we have 
	\[
	f_0=\sum_{i=1}^n\Delta_{\tta_{x_i}} (g^{(i)}_{0}) \ \text{ and } \
	f_{[d]_G,j}=\sum_{i=1}^n\Delta_{\tta_{x_i}}(g^{(i)}_{[d]_G,j}),
	\]
	which are $(\tta_{x_1},\ldots,\tta_{x_n})$-summable. Hence the lemma follows.
\end{proof}

\subsection{Reducing to simple fractions}
By Lemma~\ref{LEM:red1}, we have reduced Problem~\ref{PROB:summabilityproblem} to that for rational functions in $\bE[x_1]$ and $V_{[d]_G,j}$. The latter one is of the form
\begin{equation}\label{EQ:red1}
	f = \sum_\theta\frac{a_{\theta}}{\theta(d)^j},
\end{equation}
where $j\in\N^+$, $\theta\in G$, $a_\theta\in \bE[x_1]$, $d\in \bF[\vx]$ with $\deg_{x_1}(a_\theta)<\deg_{x_1}(d)$, and $d$
is irreducible in $\bF[\vx]$.

Let $\theta$ be an automorphism of $\bF(\vx)$, $c\in\bF^*$ and $a,b\in \bF(\vx)$. 
By expanding the expression~$c\cdot\theta(g) - g$, 
one can verify that the following reduction formula holds:
\begin{equation}\label{EQ:redformula}
	\frac{a}{\theta^\ell(b)} = c\cdot\theta(g) - g + \frac{c^{-\ell}\theta^{-\ell}(a)}{b},
\end{equation}
where $g = \sum_{i=0}^{\ell-1} \frac{c^{i-\ell}\theta^{i-\ell}(a)}{\theta^i(b)}$ if $\ell > 0$ and $g = -\sum_{i=0}^{-\ell-1} \frac{c^i\theta^{i}(a)}{\theta^{\ell+i}(b)}$ if~$\ell<0$.
For any $\theta\in G$, write $\theta=\theta_{x_1}^{a_1}\cdots\theta_{x_n}^{a_n}$ with $a_i\in\Z$. Let $\alpha_\theta=(a_1,\ldots,a_n)\in\Z^n$ be the index vector related to $\theta$.
Using this notation, we apply the above reduction formula~\eqref{EQ:redformula}  with $(c,\theta)=(c_1,\theta_{x_1}),\ldots,(c_n,\theta_{x_n})$ iteratively and then arrive at the following decomposition.

\begin{lem}\label{LEM:red2}
	Let $f\in V_{[d]_G,j}$ be given in the form~\eqref{EQ:red1} and $\mathbf{c}=(c_1,\ldots,c_n)\in (\bF^*)^n$. Then we can decompose $f$ into the form
	\begin{equation}\label{EQ:red2}
		f=\sum_{i=1}^n\Delta_{c_i\theta_{x_i}}(g_i)+r \text{ with } r=\frac{a}{d^j},
	\end{equation}
	where $g_i\in\bF(\vx)$ and $a=\sum_\theta \mathbf{c}^{-\alpha_\theta}\theta^{-1}(a_\theta)$ with $\deg_{x_1}(a)<\deg_{x_1}(d)$. Note that here we use $\mathbf{c}^{-\alpha_\theta}$ to denote~$c_1^{-a_1}\cdots c_n^{-a_n}$.
	In particular, $f$ is $(\tta_{x_1},\ldots,\tta_{x_n})$-summable if and only if $r$ is $(\tta_{x_1},\ldots,\tta_{x_n})$-summable.
\end{lem}
Next, we shall discuss the summability problem of polynomials in $
\bE[x_1]$ and of simple fractions 
\begin{equation*}
	f=\frac{a}{d^j},
\end{equation*}
where $j\in\N^+$, $a\in \bE[x_1]$, $d\in \bF[\vx]$ with $\deg_{x_1}(a)<\deg_{x_1}(d)$ and $d$
is irreducible and special with respect to $\theta_{x_1}$. When the denominator $d$ is normal with respect to $\theta_{x_1}$, the problem is nontrivial and will be solved in Sections \ref{SEC:criteria} and~\ref{SEC:tranformation}.

When $\theta_{x_1}=\si_{x_1}$, every irreducible polynomial in $x_1$ over $\bE$ is normal with respect to $\theta_{x_1}$ and every $f\in\bE[x_1]$ is $c_1\si_{x_1}$-summable for any~$c_1\in\bF^*$.

When $\theta_{x_1}=\tau_{q,x_1}$, the only special and irreducible polynomial in $\bE[x_1]$ with respect to $\theta_{x_1}$ is of the form $c\cdot x_1$ for some $c\in \bE\setminus\{0\}$. So it is sufficient to consider $f\in\bE[x_1,x_1^{-1}]$. For $c_1\in\bF^*$, note that $c_1(qx_1)^j-x_1^j=(c_1q^{j}-1)x_1^j$.
If $c_1\neq q^{-\nu}$ for any integer $\nu\in\Z$, then $f$ is $c_1\tau_{q,x_1}$-summable. 
In this case, for each~$j \in \bZ$ we have 
\begin{equation}\label{EQ:q-sumpoly}
	x_1^j = c_1 \cdot \tau_{q,x_1} \Bigl(\frac{x_1^j}{c_1q^{j}-1}\Bigr) - \frac{x_1^j}{c_1q^{j}-1}.
\end{equation}
Now suppose $c_1=q^{-\nu}$ for some $\nu\in\Z$. For each $j\in\Z$, we define $W_j$ as the $\bE$-vector subspace of $\bF(\vx)$ generated by $x_1^j$, i.e., $W_j=\{g \cdot x_1^j \mid g\in \bE\}$. Then $\bE[x_1,x_1^{-1}]$ is the direct sum of $W_j$ where $j$ ranges through $\Z$. We can write
\begin{equation}\label{EQ:qsumspecial}
	f= c_1 \cdot \tau_{q,x_1}(g)-g+f_\nu
\end{equation}
for some $g\in\bE[x_1,x_1^{-1}]$ and $f_\nu\in W_\nu$. Then $f$ is $(c_1\tau_{q,x_1},\tta_{x_2},\ldots,\tta_{x_n})$-summable if and only if $f_\nu$ is $(c_1\tau_{q,x_1},\tta_{x_2},\ldots,\tta_{x_n})$-summable provided that~$c_1 = q^{-\nu}$.
The following result shows that the summability of $f_\nu$ can be further reduced to that in $n-1$ variables.

\begin{prop}\label{PROP:laurentpoly}
	Let $f$ be in the form of~\eqref{EQ:qsumspecial}. Then we have that $f$ is $(q^{-\nu}\tau_{q,x_1},\tta_{x_2},\ldots,\tta_{x_n})$-summable in $\bF(\vx)$ if and only if $f_\nu$ is $(\tta_{x_2},\ldots,\tta_{x_n})$-summable in $\bF(\vx)$.
\end{prop}
\begin{proof}
	The sufficiency follows from the definition of summability. 
	Conversely, let $\tta_{x_1}:=q^{-\nu}\tau_{q,x_1}$ and $G:=\la \tau_{q,x_1},\theta_{x_2},\ldots,\theta_{x_n} \ra$.
	Since $f_\nu \in W_\nu\subseteq\bE[x_1,x_1^{-1}]$ and $ \bE[x_1,x_1^{-1}]$ is the union of $\bE[x_1]$ and $V_{[x_1]_G,j}$ for all $j\in\N^+$, we can assume that $f_\nu=\sum_{i=1}^n \Delta_{\tta_{x_i}}(g_i)$ with $g_i \in \bE[x_1,x_1^{-1}]$ by Lemma~\ref{LEM:red1}. Furthermore, $\bE[x_1,x_1^{-1}]=\bigoplus_{j\in\Z} W_j$ and for each $j\in\Z$, the subspace $W_j$ is also $G$-invariant, i.e., any $h\in W_j$ and $\theta\in G$ implies $\theta(h)\in W_j$. Let $g_{i,\nu}$ be the component of $g_i$ in $W_\nu$. By the similar discussion as in Lemma~\ref{LEM:red1}, we have  
	\[f_\nu=\Delta_{\tta_{x_1}}(g_{1,\nu})+\Delta_{\tta_{x_2}}(g_{2,\nu})+\cdots+\Delta_{\tta_{x_n}}(g_{n,\nu}).\]
	Here $g_{1,\nu}$ is of the form $h\cdot x_1^\nu$ for some $h\in \bE$. Consequently, the first term vanishes: 
	\[\Delta_{\tta_{x_1}}(g_{1,\nu})=h q^{-\nu}\tau_{q,x_1}(x_1^\nu)-h x_1^\nu=0.\]
	Hence $f_\nu$ is $(\tta_{x_2},\ldots,\tta_{x_n})$-summable in $\bF(\vx)$. 
\end{proof}
In Proposition~\ref{PROP:laurentpoly}, if $f_\nu$ is $(\tta_{x_2},\ldots,\tta_{x_{n}})$-summable in $\bF(\vx)$, then $f_\nu$ can be regarded as $(\tta_{x_2},\ldots,\tta_{x_{n}})$-summable in $\bK(x_2,\ldots,x_{n})$ with $\bK=\bF(x_1)$, since $\bF^*\subseteq\bK^*$. The latter one is contained in the summability problem in $n-1$ variables (if we replace $\bF$ by $\bK$).
\section{Isotropy groups}\label{SEC:iso}
Consider the group action given in Section~\ref{SEC:adddecomp}.
Let $p\in\bF[\vx]$ be a nonconstant polynomial and $K$ be a subgroup of $G=\la\theta_{x_1},\ldots,\theta_{x_n}\ra$. The set
\begin{equation*}
	K_p:=\{\theta\in K\mid \theta(p)=c\cdot p \text{ for some } c\in \bF^*\}
\end{equation*}
is a subgroup of $K$, called the {\em isotropy group} of $p$ in $K$. If two polynomials $p,q$ are $K$-equivalent, then $K_p=K_q$. In this section, we shall discuss some algebraic properties of $G_p$, which will be used in Section~\ref{SEC:criteria}.

Let $G^\si=\la\si_{y_1},\ldots,\si_{y_k}\ra$ and $G^\tau=\la\tau_{z_1},\ldots,\tau_{z_m}\ra$ be subgroups of $G$. Then $G=G^{\si}\oplus G^\tau$. The isotropy groups of $p$ in $G^{\si}$ and $G^{\tau}$ are denoted by $G^\si_p$ and $G^{\tau}_p$, respectively. Every subgroup of $G$ is a free abelian group. Furthermore, the following structure property of the quotient group $G^\si/G^\si_p$ is given by Sato~\cite[Lemma~A-3]{Sato1990}.

\begin{lem}\label{LEM:G1/G1pfree}
	$G^\si/G^\si_p$ is a free abelian group.
\end{lem}

Similarly, we have the following lemma in $q$-shift case.
\begin{lem}\label{LEM:G2/G2pfree} 
	$G^\tau/G^\tau_p$ is a free abelian group.
\end{lem}
\begin{proof}
	By~\cite[Chapter~\uppercase\expandafter{\romannumeral3}, Theorem~7.3]{BookLang}, it suffices to show $G^\tau/G^\tau_p$ is torsion free. Suppose $\tau_0^{\ell}\in G^\tau_p$ for some $\ell>0$. Write $\tau_0=\tau_{q,z_1}^{t_1}\cdots\tau_{q,z_m}^{t_m}$ and $p=\sum a_I \vz^I$ with $I=(i_1,\ldots,i_m)\in\Z^m$, $\vz^I=z_1^{i_1}\cdots z_m ^{i_m}$ and  $a_I\in\bF[\vy]$. Let $T$ be the set of all monomials $\vz^I$ appearing in $p$ with nonzero coefficients in $\bF[\vy]$. Then $\tau_0^\ell(p)=c\cdot p$ implies $c=q^{\ell_0}$ for some $\ell_0\in \Z$ and we have
	\[\sum_Ia_I q^{\ell( t_1i_1+\cdots+t_mi_m)-\ell_0}\vz^I=\sum_I a_I \vz^I.\]
	It follows that for any $\vz^I\in T$, \[\ell(t_1i_1+\cdots+t_mi_m)=\ell_0,\]
	since $q$ is not a root of unity. So $\ell$ divides $\ell_0$ and let $\tilde c=q^{\ell_0/\ell}$. Then $\tau_0(p)=\tilde c\cdot p$, i.e., $\tau_0\in G^\tau_p$.
\end{proof}
\begin{prop}\label{PROP:G/Gpfree}
	$G_p=G_p^\si\oplus G_p^\tau $. Therefore, $G/G_p\cong G^\si/G^\si_p \oplus G^\tau/G^\tau_p$ is a free abelian group.
\end{prop}
\begin{proof}
	Since $G=G^\si \oplus G^\tau$, we only need to show that if $\theta=\si_0\cdot \tau_0\in G_p$ with $\si_0\in G^\si, \tau_0\in G^\tau$, then $\si_0\in G_p^\si, \tau_0\in G^\tau_p$. Write $p=\sum f_I(\vy)\vz^I$ with $f_I(\vy)\in\bF[\vy]$ and $T=\{\vz^I \mid f_I(\vy)\neq 0\}$. If $\theta (p)=c\cdot p$ for some $c\in \bF^*$, then
	\[\sum_I \si_0(f_I(\vy))\tau_0(\vz^I)=\sum_I c \cdot f_I(\vy)\vz^I\]
	Since $\si_0$ preserves the leading coefficient and $\tau_0$ preserves the term structure, we have $\si_0(f_I(\vy))=f_I(\vy)$ and $\tau_0(\vz^I)=c\cdot\vz^I$ for any $\vz^I\in T$. Hence $\si_0(p)=p$ and $\tau_0(p)=c\cdot p$.
\end{proof}
If $n>1$, let $H=\la\theta_{x_1},\ldots,\theta_{x_{n-1}}\ra$ be a subgroup of $G$ and $H_p$ be the isotropy group of $p$ in $H$. By Proposition~\ref{PROP:G/Gpfree}, both $G/G_p$ and $H/H_p$ are free abelian groups. If $p$ is normal with respect to $\theta_{x_1}$ and of positive degree in $x_1$, the ranks of $G_p$ and $H_p$ are strictly less than that of $G$ and $H$ respectively.
\begin{rem}
	Computing a basis of $G_p$ can be reduced to solving systems of linear Diophantine equations. By the direct sum decomposition of $G_p$, we can compute bases for $G_p^{\si}$ and $G_p^{\tau}$ separately. A defining set of linear equations for the basis of $G_p^\sigma$ can be derived by utilizing methods from shift equivalence testing, see~\cite{Dvir2014} and~\cite[Section~3]{ChenDuFang2025}.
	The basis of $G_p^\tau$ can be obtained via $q$-shift equivalence testing~\cite[Theorem~1]{Wang2021}.
\end{rem}
The following lemma is a direct extension of Lemma~4.4 of~\cite{ChenDuFang2025} which can be proved in the same way.  
\begin{lem}\label{LEM:Gd/Hdfree}
	$G_p/H_p$ is a free abelian group with $\rank G_p/H_p\leq1$.
\end{lem}

\section{Summability criteria}\label{SEC:criteria}
Combining Lemma~\ref{LEM:red1} and Lemma~\ref{LEM:red2}, we can reduce Problem~\ref{PROB:summabilityproblem} to that for simple fractions
\begin{equation}\label{EQ:simplefraction}
	f = \frac{a}{d^j},
\end{equation}
where $j\in\N^+$, $a\in \bE[x_1]$, $d\in \bF[\vx]$ with $\deg_{x_1}(a)<\deg_{x_1}(d)$ and $d$
is irreducible and normal with respect to $\theta_{x_1}$.
In this section, we shall present a summability criterion for such simple fractions.

For the univariate case with $n=1$, the problem is for $f=a/d^j$ with $a,d\in \bF[x_1]$ in the form~\eqref{EQ:simplefraction} to decide whether $f$ is $\tta_{x_1}$-summable. Since $d$ is normal with respect to $\theta_{x_1}$, we get the following criterion of $\tta_{x_1}$-summability from~\cite[Lemma~6.3]{Hardouin2008}; see also~\cite{Abramov1975,Paule1995b,ChenSinger2012}.
\begin{lem}\label{LEM:unicase}
	Let $f=a/d^j$ be of the form~\eqref{EQ:simplefraction} with $a,d\in\bF[x_1]$ and $\tta_{x_1}=c\cdot \theta_{x_1}$ for some $c\in\bF^*$. Then $f$ is $\tta_{x_1}$-summable if and only if $a=0$.
\end{lem}
For the multivariate rational functions with $n>1$, we proceed in the following two cases.
\begin{enumerate}
	\item $\rank G_d/H_d=0,$
	\item $\rank G_d/H_d=1.$
\end{enumerate}

Firstly if $\rank G_d/H_d=0$, the summability problem in $n$ variables can be reduced to that in $n-1$ variables by the following lemma which is a direct extension of Lemma~5.6 in~\cite{ChenDuFang2025}.
Note that the $(\tta_{x_1},\ldots,\tta_{x_{n-1}})$-summability in $\bF(\vx)$ implies the $(\tta_{x_1},\ldots,\tta_{x_{n-1}})$-summability in $\bK(x_1,\ldots,x_{n-1})$ with $\bK=\bF(x_n)$. Furthermore, for this reason, we have the corresponding version of Lemmas \ref{LEM:G-invariant}, ~\ref{LEM:red1} and~\ref{LEM:red2} for $(\tta_{x_1},\ldots,\tta_{x_{n-1}})$-summability in $\bF(\vx)$, which will be used in the proof of Lemma~\ref{LEM:rank=0} and Theorem~\ref{THM:rank=1}.
\begin{lem}\label{LEM:rank=0}
	Let $f=a/d^j\in \F(\vx)$ be given in the form~\eqref{EQ:simplefraction}. If~$
	\rank G_d/H_d=0$ and~$n>1$, then $f$ is $(\tta_{x_1},\ldots,\tta_{x_n})$-summable in $\bF(\vx)$ if and only if $f$ is $(\tta_{x_1},\ldots,\tta_{x_{n-1}})$-summable in $\bF(\vx)$.
\end{lem}
\begin{proof}
	The sufficiency follows from the definition of summability. For the necessity, suppose $f$ is $(\tta_{x_1},\ldots,\tta_{x_n})$-summable in $\bF(\vx)$ but not $(\tta_{x_1},\ldots,\tta_{x_{n-1}})$-summable in $\bF(\vx)$. By Lemma~\ref{LEM:red1}, we get
	$f=\Delta_{\tta_{x_1}}(g_1)+\cdots+\Delta_{\tta_{x_n}}(g_n)$
	with $g_1,\ldots,g_n$ in the same subspace $V_{[d]_G,j}$ as $f$. For each $i$ with $1\leq i\leq n$, write $\tta_{x_i}=c_i\cdot\theta_{x_i}$ for some $c_i\in\bF^*$. Applying the reduction formula~\eqref{EQ:redformula}  with $(c,\theta)=(c_1,\theta_{x_1}),\ldots,(c_{n-1},\theta_{x_{n-1}})$ iteratively, we can decompose $g_n$ as
	\begin{equation*}
		g_n=\sum_{i=1}^{n-1}\Delta_{\tta_{x_i}}(u_i)+\sum_{\ell=0}^{\rho}\frac{\lam_\ell}{\theta_{x_n}^{\ell}(\mu)^j},
	\end{equation*}
	where $u_i\in \F(\vx)$, $\rho\in\N$, $\lam_\ell\in\F(\hx_1)[x_1]$, $\mu\in \F[\vx]$ with $\deg_{x_1}(\lam_{\ell})<\deg_{x_1}(d)$ and $\mu$ is in the same $G$-orbit as $d$. Furthermore, we can assume $\lam_{0}\lam_{\rho}\neq0$ and each nonzero $\lam_\ell/\theta_{x_n}^{\ell}(\mu)^j$ is not $(\tta_{x_1},\ldots,\tta_{x_{n-1}})$-summable in $\bF(\vx)$. Substituting $g_n$ into the original equation for $f$ and using the commutativity between $\theta_{x_n}$ and $\theta_{x_i}$($1\leq i\leq n-1$), we obtain that
	\[ f-\Delta_{\tta_{x_n}}\biggl(\sum_{\ell=0}^{\rho}\frac{\lam_\ell}{\theta_{x_n}^{\ell}(\mu)^j}\biggr)=\sum_{i=1}^{n-1}\Delta_{\tta_{x_i}}(h_i),\]
	where  $h_i=g_i+\Delta_{\tta_{x_n}}(u_i)$ for all $1\leq i\leq n-1$.
	Then
	\begin{equation}\label{EQ:rk=0_red}
		f+\sum_{\ell=0}^{\rho+1}\frac{\tilde{\lam}_\ell}{\theta_{x_n}^{\ell}(\mu)^j}=\sum_{i=1}^{n-1}\Delta_{\tta_{x_i}}(h_i),
	\end{equation}
	where $\tilde{\lam}_{0}=\lam_0$, $\tilde{\lam}_{\rho+1}=-c_n\theta_{x_n}(\lam_\rho)$, $\tilde{\lam}_\ell=\lam_\ell-c_n\theta_{x_n}(\lam_{\ell-1})$ for all $1\leq\ell\leq\rho$.
	Since $d$ and $\mu$ are in the same $G$-orbit, it follows that $G_d=G_\mu$ and $H_d=H_\mu$. Since~$\rank G_d/H_d=0$, we have all $\theta_{x_n}^\ell(\mu)$ with $\ell\in \Z$ are in distinct $H$-orbits. In particular, the $H$-orbits $[\mu]_H,[\theta_{x_n}(\mu)]_H,\ldots,[\theta_{x_n}^{\rho+1}(\mu)]_H$ are distinct. On the other hand, the left hand side of~\eqref{EQ:rk=0_red} is $(\tta_{x_1},\ldots,\tta_{x_{n-1}})$-summable. However, $\tilde{\lam}_{0}/\mu^j={\lam}_{0}/\mu^j$ is not $(\tta_{x_1},\ldots,\tta_{x_{n-1}})$-summable. By Lemma~\ref{LEM:red1} (in $n-1$ variables), the only possible choice is that $\mu$ lies in the same $H$-orbit as $d$, denoted by $\mu\sim_H d$. By the similar discussion, we have $\theta_{x_n}^{\rho+1}(\mu)\sim_H d$ since~$\tilde \lambda_{\rho +1}\neq 0$. Hence $\mu\sim_H \theta_{x_n}^{\rho+1}(\mu)$, which leads to a contradiction since $\rho$ is a nonnegative integer and~$\rank G_{\mu}/H_{\mu} = 0$. The lemma follows.
\end{proof}
We are now ready to state and prove the main theorem throughout this paper. In the pure difference case, it coincides with Theorem~5.9 in~\cite{ChenDuFang2025}.
\begin{thm}\label{THM:rank=1}
	Let $f=a/d^j\in \F(\vx)$ be of the form~\eqref{EQ:simplefraction}. Let $\{\theta_i\}_{i=1}^r(1\leq r< n)$ be a basis of $G_d$ (take $\theta_1={\bf 1}$ if $G_d=\{{\bf 1}\}$). Suppose $\theta_i=\theta_{\vx}^{\alpha_i}$ with $\alpha_i\in\Z^n$ and $\theta_i(d)=\e_id$ for some $\e_i\in\bF^*$. Then for any $\mathbf{c}=(c_1,\ldots,c_n)\in(\bF^*)^n$, $f$ is $(c_1\theta_{x_1},\ldots,c_n\theta_{x_n})$-summable in $\bF(\vx)$ if and only if $a=\sum_{i=1}^{r}\Delta_{\tta_i}(b_i)$ for some $b_i\in \F(\hx_1)[x_1]$ with $\tta_i=\e_i^{-j}\mathbf{c}^{\alpha_i}\theta_i$ and $\deg_{x_1}(b_i)<\deg_{x_1}(d)$ for all $1\leq i\leq r$.
\end{thm}
The above theorem reduces the number of difference operators in the summability problem. Before proving it, we first show that this conclusion is preserved by any basis exchange of $G_d$ in Lemma~\ref{LEM:exchangebases}, so that it is sufficient to prove Theorem~\ref{THM:rank=1} for a special basis. Next, we use induction and organize the proof in two cases according to $\rank G_d/H_d=0$ or $1$.

\begin{lem}\label{LEM:exchangebases}
	Let $\{\omega_i\}_{i=1}^r$, $\{\eta_i\}_{i=1}^r(r\geq 1)$ be two bases of $G_d$ such that for each $1\leq i\leq r$, $\omega_i(d)=\varepsilon_i d$, $ \eta_i(d)=e_i d$ for some $\varepsilon_i, e_i\in \bF^*$. Suppose~$\omega_i=\theta_{\vx}^{\alpha_i}$, $ \eta_i=\theta_{\vx}^{\beta_i}$ with $\alpha_i, \beta_i\in\Z^n$. Then for any $f\in \bF(\vx)$, $\mathbf{c}\in(\bF^*)^n$ and $s\in\Z$, $f$ is $(\varepsilon^s_1\mathbf{c}^{\alpha_1}\omega_1,\ldots,\varepsilon^s_r\mathbf{c}^{\alpha_r}\omega_r)$-summable in $\bF(\vx)$ if and only if $f$ is $(e^s_1\mathbf{c}^{\beta_1}\eta_1,\ldots,e^s_r\mathbf{c}^{\beta_r}\eta_r)$-summable in $\bF(\vx)$.
\end{lem}
\begin{proof}
	By the symmetry of $\{\omega_i\}_{i=1}^r$ and $\{\eta_i\}_{i=1}^r$, we only need to show one direction. For the sufficiency, it is enough to prove that if $\eta=\theta_{\vx}^{\beta}\in\la\omega_1,\ldots,\omega_r\ra$ with $\eta(d)=ed$ for some $e\in\bF^*$, then
	\[e^s\mathbf{c}^\beta\eta-{\bf 1}=(\e^s_1c^{\alpha_1}_1\omega_1-{\bf 1})\bar{\omega}_1+\cdots+(\e^s_rc^{\alpha_r}_r\omega_r-{\bf 1})\bar{\omega}_r\]
	for some $\bar{\omega}_1,\ldots,\bar{\omega}_r\in \bF[G]$. If $\eta=\omega_1^{t_1}\cdots\omega_r^{t_r}$ with $t_i\in\Z$, then $\beta=\alpha_1t_1+\cdots+\alpha_rt_r$ and $e=\e_1^{t_1}\ldots\e_r^{t_r}$. So we have
	\begin{align*}
		e^s\mathbf{c}^\beta\eta-{\bf 1}&=\e_1^{st_1}\cdots\e_r^{st_r}\mathbf{c}^{\alpha_1t_1}\cdots \mathbf{c}^{\alpha_rt_r}\omega_1^{t_1}\cdots\omega_r^{t_r}-\one  \\
		&=(\e^s_1\mathbf{c}^{\a_1}\w_1)^{t_1}\cdots(\e^s_r\mathbf{c}^{\a_r}\w_r)^{t_r}-\one\\
		&=(\e^s_1\mathbf{c}^{\a_1}\w_1- \one)\bar{\w}_1+\cdots+(\e^s_r\mathbf{c}^{\a_r}\w_r-\one)\bar{\w}_r
	\end{align*}
	for some $\bar{\omega}_1,\ldots,\bar{\omega}_r\in \bF[G]$. The last equality is obtained from the identity 
	\[
		x_1^{t_1}\cdots x_r^{t_r}-1=(x_1-1)g_1+\cdots(x_r-1)g_r
	\]
	 for some $g_1,\ldots,g_r\in\bF[x_1,\ldots,x_r,x_1^{-1},\ldots,x_r^{-1}]$, which can be proved by induction on $r$.
\end{proof}

\begin{proof}[Proof of Theorem~\ref{THM:rank=1}]
	For the sufficiency, since $\theta_i(d)=\e_id$, we have
	\begin{align*}
		\frac{a}{d^j} & = \frac{\e_1^{-j}\mathbf{c}^{\a_1}\theta_1(b_1)}{\e_1^{-j}\theta_1(d)^j}-\frac{b_1}{d^j}+\cdots+\frac{\e_r^{-j}\mathbf{c}^{\a_r}\theta_r(b_r)}{\e_r^{-j}\theta_r(d)^j}-\frac{b_r}{d^j}\\
		& =\mathbf{c}^{\a_1}\theta_1\bigl(\frac{b_1}{d^j}\bigr)-\frac{b_1}{d^j}+\cdots+\mathbf{c}^{\a_r}\theta_r\bigl(\frac{b_r}{d^j}\bigr)-\frac{b_r}{d^j}.
	\end{align*}
	By the formula~\eqref{EQ:red2}, we have that $ a/d^j=\sum_{i=1}^n\Delta_{c_i\theta_{x_i}}(g_i)$ for some~$g_i\in\bF(\vx)$. 
	
	For the necessity, we proceed by induction on the number of ($q$-)difference operators $n$.
	We begin the induction with $n=1$. The problem is to decide the $c_1\theta_{x_1}$-summability of $f=a/d^j\in\bF(x_1,\ldots,x_n)$. Then the subgroup $G_d$ of~$G=\la \theta_{x_1} \ra$ is trivial and $\theta_1=\one$, since $d$ is normal with respect to~$\theta_{x_1}$ and of positive degree in $x_1$. In this case, $\e_1=1,\a_1=(0),\mathbf{c}^{\a_1}=1$, and $f$ is $c_1\theta_{x_1}$-summable in $\bF(x_1,\ldots,x_n)$ if and only if $a=0$ by Lemma~\ref{LEM:unicase} (replace $\bF$ by $\bE$). Now let us assume $n>1$ and formulate the inductive hypothesis for $n-1$:
	
	{\em If $\{\w_i\}_{i=1}^{s}$ is a basis of $H_d$ and $\w_i=\theta_{\vx}^{\beta_i}$, $\w_i(d)=e_id$ for some $\beta_i\in\Z^n$ and~$e_i\in\bF^*$, then $f$ is $(c_1\theta_{x_1},\ldots,c_{n-1}\theta_{x_{n-1}})$-summable in $\bF(x_1,\ldots,x_n)$ if and only if $a=\sum_{i=1}^{s}\Delta_{\tilde{\w}_i}(b_i)$ for some $b_i\in \F(\hx_1)[x_1]$ with $\tilde{\w}_i=e_i^{-j}\mathbf{c}^{\beta_i}\w_i$ and $\deg_{x_1}(b_i)<\deg_{x_1}(d)$ for all $1\leq i\leq s$.}
	
	We shall proceed by a case-by-case analysis according to the rank of $ G_d/H_d$.
	If $\rank G_d/H_d=0$, then $H_d=G_d$. Lemma~\ref{LEM:rank=0} implies that $f$ is
	$(c_1\theta_{x_1},\ldots,c_n\theta_{x_n})$-summable in $
	\bF(\vx)$ if and only if $f$ is $(c_1\theta_{x_1},\ldots,c_{n-1}\theta_{x_{n-1}})$-summable in~$\bF(\vx)$.
	So the assertion is true by the inductive hypothesis.
	If $\rank G_d/H_d=1$, using Lemma~\ref{LEM:exchangebases}, we may assume $\{\theta_i\}_{i=1}^{r}$ is a basis of $G_d$ such that
	$H_d=\la\theta_1,\ldots,\theta_{r-1}\ra$ and $G_d/H_d=\la\bar{\theta}_r\ra$. (If $r=1$, then $H_d=\{\one\}$.)
	By the argument of Lemma~\ref{LEM:Gd/Hdfree}, we can write $\theta_r=\theta_{\vx}^{\alpha_r}$ as $\theta_{x_1}^{-t_1}\cdots\theta_{x_{n-1}}^{-t_{n-1}}\theta_{x_n}^{t_n}$ with~$t_n$ being the smallest positive integer among all elements in $G_d$. By Lemma~\ref{LEM:red1}, we can assume $f=\Delta_{c_1\theta_{x_1}}(g_1)+\cdots+\Delta_{c_n\theta_{x_n}}(g_n)$ with~$g_i\in V_{[d]_G,j}$. Here $g_n$ can be decomposed as
	\begin{equation*}
		g_n=\sum_{i=1}^{n-1}\Delta_{c_i\theta_{x_i}}(u_i)+\sum_{\ell=0}^{t_n-1}\frac{\lam_\ell}{\theta_{x_n}^{\ell}(d)^j},
	\end{equation*}
	where $u_i\in \F(\vx)$ and $\lam_\ell\in\F(\hx_1)[x_1]$ with $\deg_{x_1}(\lam_{\ell})<\deg_{x_1}(d)$. Then we have
	\begin{equation}\label{EQ:rk=1_red}
		f-\Delta_{c_n\theta_{x_n}}\biggl(\sum_{\ell=0}^{t_n-1}\frac{\lam_\ell}{\theta_{x_n}^{\ell}(d)^j}\biggr)=\sum_{i=1}^{n-1}\Delta_{c_i\theta_{x_i}}(h_i),
	\end{equation}
	where $h_i=g_i+\Delta_{c_n\theta_{x_n}}(u_i)$.
	Note that $\theta_{x_n}^{t_n}(d)=\e_r\theta_{x_1}^{t_1}\cdots\theta_{x_{n-1}}^{t_{n-1}}(d)$.
	Applying the reduction formula~\eqref{EQ:redformula} to simplify~\eqref{EQ:rk=1_red}, we get
	\begin{equation}\label{EQ:rk=1_red2}
		\tilde{f}:=\sum_{\ell=0}^{t_n-1}\frac{\tilde{\lam}_\ell}{\theta_{x_n}^{\ell}(d)^j}=\sum_{i=1}^{n-1}\Delta_{c_i\theta_{x_i}}(\tilde{h}_i)
	\end{equation}
	for some rational functions $\tilde{h}_i\in\F(\vx)$, where $\tilde{\lam}_\ell=\lam_\ell-c_n\theta_{x_n}(\lam_{\ell-1})$ for $1\leq\ell\leq t_n-1$ and 
	\[\tilde{\lam}_0=a+\lam_{0}-\e_r^{-j}(c_1\theta_{x_1})^{-t_1}\cdots(c_{n-1}\theta_{x_{n-1}})^{-t_{n-1}}c_n\theta_{x_n}(\lam_{t_{n}-1}).\]
	This means that $\tilde{f}$ is $(c_1\theta_{x_1},\ldots,c_{n-1}\theta_{x_{n-1}})$-summable in $\bF(\vx)$. Notice that $[d]_H,[\theta_{x_n}(d)]_H,\ldots,[\theta_{x_n}^{t_n-1}(d)]_H$ are distinct $H$-orbits due to the minimality of $t_n$. By Lemma~\ref{LEM:red1}, it implies that each $\tilde{\lambda}_\ell/\theta_{x_n}^{\ell}(d)^j$ is $(c_1\theta_{x_1},\ldots,c_{n-1}\theta_{x_{n-1}})$-summable in $\bF(\vx)$ for~$0\leq \ell\leq t_n-1$. Let $W$ denote the $\bF$-vector subspace of $\bF(\vx)$ consisting of all elements in the form of $\sum_{i=1}^{r-1}\Delta_{\tta_i}(b_i)$ with $\tta_i=\e_i^{-j}\mathbf{c}^{\alpha_i}\theta_i$, $b_i\in \bE[x_1]$, and $\deg_{x_1}(b_i)<\deg_{x_1}(d)$ for~$1\leq i\leq r-1$. (If $r=1$, take $W=\{0\}$.) Since $H_{d}=H_{\theta_{x_n}^\ell(d)}$ for $\ell=0,\ldots,t_n-1$, we apply the inductive hypothesis for each rational function~$\tilde{\lambda}_\ell/\theta_{x_n}^{\ell}(d)^j$ and conclude that
	\begin{equation*}
		\left\{ \begin{array}{l}
			0 \equiv a + \lam_{0}-\e_r^{-j}(c_1\theta_{x_1})^{-t_1}\cdots(c_{n-1}\theta_{x_{n-1}})^{-t_{n-1}}c_n\theta_{x_n}(\lam_{t_{n}-1}),\\
			0 \equiv \lambda_1 - c_n\theta_{x_n}(\lambda_{0}),\\
			\qquad\qquad\vdots\\
			0 \equiv \lambda_{t_n-1} - c_n\theta_{x_n}(\lambda_{t_n-2}) ,
		\end{array} \right.
	\end{equation*}
	where $\equiv$ means the congruence relation modulo~$  W$.
	Since $W$ is $G$-invariant, the above system of congruences leads to
	\begin{align*}
		a & \equiv\e_r^{-j}(c_1\theta_{x_1})^{-t_1}\cdots(c_{n-1}\theta_{x_{n-1}})^{-t_{n-1}}(c_n\theta_{x_n})^{t_n}(\lam_{0})-\lambda_0 \\
		& \equiv \e_r^{-j}\mathbf{c}^{\alpha_r}\theta_r(\lambda_0)-\lambda_0.
	\end{align*}
	This completes the proof.
\end{proof}

\begin{rem}\label{REM:mixeda}
	By Proposition~\ref{PROP:G/Gpfree}, we get $G_d=G_d^{\si}\oplus G_d^\tau$, so there exists a basis $\{\si_1,\ldots,\si_{r_1},\tau_1,\ldots,\tau_{r_2}\}$ of $G_d$ such that $\{\si_1,\ldots,\si_{r_1}\}$, $\{\tau_1,\ldots,\tau_{r_2}\}$ are bases of $G_d^{\si}$ and $G_d^\tau$ respectively. Then $\si_i(d)=d$ and $\tau_\ell(d)=q^{\nu_\ell}d$ for some $\nu_\ell\in\Z$. By Theorem~\ref{THM:rank=1}, we obtain that $a/d^j$ is $(\si_{y_1},\ldots,\si_{y_k},\tau_{q,z_1},\ldots,\tau_{q,z_m})$-summable in~$\bF(\vx)$ if and only if there exist $b_1,\ldots,b_{r_1}\in \bE[x_1]$ with $\deg_{x_1}(b_i)<\deg_{x_1}(d)$ and $\lam_1,\ldots,\lam_{r_2}\in\bE[x_1]$ with $\deg_{x_1}(\lambda_\ell)<\deg_{x_1}(d)$ such that
	\begin{equation}\label{EQ:sum-cri-1}
		a=\sum_{i=1}^{r_1}\Delta_{\si_i}(b_i)+\sum_{\ell=1}^{r_2}\Delta_{\tilde\tau_\ell}(\lambda_\ell),
	\end{equation}
	where $\tilde\tau_\ell=q^{-j\nu_\ell}\tau_\ell$. 
	By the similar discussion as in Lemma~\ref{LEM:red1}, 
	the summability criterion for~$a/d^j$ can be refined as 
	\[a \text{ is } (\si_1,\ldots,\si_{r_1},\tilde \tau_1,\ldots, \tilde \tau_{r_2})\text{-summable in } \bF(\vx).\] 
\end{rem}
In pure difference case, Example~5.11 in~\cite{ChenDuFang2025} verifies that for positive integers $s$ and $n$, the function $f=1/(x_1^s+\cdots+x_n^s)$ is summable if and only if $s=1$ and $n>1$. We now show that $f$ is always $q$-summable in the following example, for which the bivariate case has been shown by Example~3.19 of~\cite{ChenSinger2014}.

\begin{exam}\label{EX:qsum}
	We consider the $(\tau_{q,x_1},\ldots, \tau_{q,x_n})$-summability of~$f$ in~$\bQ(x_1,\ldots,x_n)$, where
	\[
	f := \frac{1}{x_1^s+\cdots+x_n^s} \ \text{with} \ s,n \in \bN^{+}.   
	\]
	Let $G_d$ be the isotropy group of $d := x_1^s+\cdots+x_n^s$ in~$\la \tau_{q,x_1},\ldots, \tau_{q,x_n} \ra$.
	When $n=1$, applying Equation~\eqref{EQ:q-sumpoly} yields
	\[
	f = \frac{1}{x_1^s} = \Delta_{\tau_{q,x_1}}\Bigl( \frac{1/(q^{-s}-1)}{x_1^s}\Bigr).
	\]
	For $n >1$,
	it is easy to check that $G_d$ is generated by~$\tau := \tau_{q,x_1}\cdots\tau_{q,x_n} $ with~$\tau(d) = q^s d$. By Theorem~\ref{THM:rank=1}, $f$ is $(\tau_{q,x_1},\ldots, \tau_{q,x_n})$-summable if and only if $1 = q^{-s}\tau(b)-b$ for some $b \in \bQ(\hat{\vx}_1)[x_1]$ with~$\deg_{x_1}(b) < s$. Taking $b = \frac{1}{q^{-s}-1}$ satisfies the condition and then we have
	\begin{align}\label{EQ:examq}
		f =&\, \tau_{q,x_1}\cdots\tau_{q,x_n}\Bigl(\frac{b}{d}\Bigr) - \frac{b}{d} \nonumber \\
		=&\, \Delta_{\tau_{q,x_1}}\left( \tau_{q,x_2}\cdots\tau_{q,x_n} \Bigl(\frac{b}{d}\Bigr)\right)
		+ \Delta_{\tau_{q,x_2}}\left( \tau_{q,x_3}\cdots\tau_{q,x_n} \Bigl(\frac{b}{d}\Bigr)\right) \nonumber\\
		&+\cdots + \Delta_{\tau_{q,x_n}}\Bigl(\frac{b}{d}\Bigr).
	\end{align}
	Hence $f$ is $(\tau_{q,x_1},\ldots, \tau_{q,x_n})$-summable for any~$s,n \in \bN^+$.
\end{exam}

\section{Difference transformations}\label{SEC:tranformation}
In this section, we shall construct an $\bF$-endomorphism of $\bF(\vx)$ that can transfer the summability problem for general operators of $G$ into the usual case. 
It can be constructed by considering the difference case and $q$-difference case, separately.
Note that the pure difference case is established in~\cite[Proposition~5.12]{ChenDuFang2025}, which can be restated as follows.

\begin{prop}\label{PROP:diffhomo}
	Let $f\in\bK(\vy)$ with $\bK=\bF(\vz)$ and $\{\si_i\}_{i=1}^{r}(1\leq r\leq k)$ 
	be a family of independent elements in~$G^\si$. Then there exists a $\bK$-automorphism~$\phi$ of $\bK(\vy)$ such that $\phi\circ\si_i=\si_{y_i}\circ\phi$ for all $1\leq i\leq r$ and therefore $f$ is $(\si_1,\ldots,\si_{r})$-summable in $\bK(\vy)$ if and only if $\phi(f)$ is $(\si_{y_1},\ldots,\si_{y_{r}})$-summable in $\bK(\vy)$.
\end{prop}

In the $q$-difference case, we may need to extend the $q$-summability in $\bK(\vz)$ with $\bK=\bF(\vy)$ to its algebraic closure $\overline{\bK(\vz)}$, 
because the inverse of a $\bK$-automorphism may involve radical expressions. 
For each $i\in\{1,\ldots,m\}$, let $\tau_{q,z_i}$ denote an arbitrary extension of $\tau_{q,z_i}$ to $\overline{\bK}$-automorphism of $\overline{\bK(\vz)}$. The following lemma is a natural generalization of Theorem~3.2 in~\cite{ChenSinger2014} from the summability in bivariate case to that in multivariate case.
\begin{lem}\label{LEM:sumalgclosure}
	Let $f\in \bK(\vz)$ with $\bK=\bF(\vy)$ and $\{\tau_i\}_{i=1}^r$ be a family of some elements in $G^\tau$. Then $f$ is $(\tau_1,\ldots,\tau_r)$-summable in $\overline{\bK(\vz)}$ if and only if $f$ is $(\tau_1,\ldots,\tau_r)$-summable in $\bK(\vz)$.
\end{lem}
Now we introduce some notation in $\overline{\bK(\vz)}$, which will be used in Proposition~\ref{PROP:qdiffhomo}.
For any $h \in \{q, z_1,\ldots,z_m \}$, we fix once and for all a compatible system
\(\{h^{1/t}\}_{t\in\N^+}\subseteq \overline{\bK(\vz)}\) of roots, where
\(h^{1/t}\) is a root of \(X^t-h\) and $(h^{1/t_1})^s=h^{1/t_2}$
if $t_1=s\cdot t_2,\ s\in\N^+$.
Such a compatible system exists since one can first fix~$\{a_n\}_{n \in \bN^+} \subseteq \overline{\bK(\vz)}$ such that $a_1 = h$ and~$a_{n+1}^{n+1} = a_n$. 
For any~$t \in \bN^+$, choose a sufficiently large~$n$ such that~$t \mid n!$, and then set~$h^{1/t} := a_n^{n!/t}$.
Based on this compatible root system~$\{h^{1/t}\}_{t\in\N^+}$, let~$ h^{1/(-t)}  := 1/h^{1/t}$.
For $s, t, \ell, \rho \in \bZ$ with~$t, \rho \neq 0$, set 
\[
h^{s/t}:=(h^{1/t})^s \quad \text{and} \quad (h^{s/t})^{\ell /\rho} := h^{s\ell /t \rho}.
\]
Hence we also have the general rules 
\[
h^{\frac{s \ell}{s \rho}}=h^{\frac{\ell}{\rho}}\quad \text{and} \quad h^{\frac{s}{t}}\cdot h^{\frac{\ell}{\rho}}=h^{\frac{s}{t}+\frac{\ell}{\rho}}.
\]
In this sense, define~$\tau_{q,z_i}(z_i^{s/t}):=q^{s/t}\cdot z_i^{s/t}$ for~$1\leq i \leq m$. 

\begin{prop}\label{PROP:qdiffhomo}
Let $f\in\bK(\vz)$ with $\bK=\bF(\vy)$ and $\{\tau_i\}_{i=1}^{r}(1\leq r\leq m)$ be a family of independent elements in $G^\tau$. Then there exists a $\bK$-endomorphism $\varphi$ of $\bK(\vz)$ such that $\varphi\circ\tau_i=\tau_{q,z_i}\circ\varphi$ for all $1\leq i\leq r$ and furthermore\ $f$ is $(\tau_1,\ldots,\tau_r)$-summable in $
\bK(\vz)$ if and only if $\varphi(f)$ is $(\tau_{q,z_1},\ldots,\tau_{q,z_r})$-summable in $\bK(\vz)$.
\end{prop}
\begin{proof}
Assume $\tau_i=\tau_{q,z_1}^{b_{i,1}}\cdots\tau_{q,z_m}^{b_{i,m}}$ with $\beta_i=(b_{i,1},\ldots,b_{i,m})$ as a vector in~$\Z^m\subseteq \Q^m$ for~$i=1,\ldots,r$. 
Then $\beta_1,\ldots,\beta_r$ are linearly independent over $\Z$, and thus also over $\Q$. Hence there exist vectors $\beta_{r+1},\ldots,\beta_m$ such that $\{\beta_1,\ldots,\beta_m\}$ forms a basis of $\Q^m$. Suppose $\beta_i=(b_{i,1},\ldots,b_{i,m})$ for $i=r+1,\ldots,m$ and then $B:=(b_{i,j})\in\Q^{m\times m}$ is an invertible matrix.
Define a $\bK$-endomorphism $\varphi:\bK(\vz) \to\bK(\vz)$ by $\varphi(\vz):=\vz^B$, which means
\[u_j:=\varphi(z_j)=z_1^{b_{1,j}}\cdots z_m^{b_{m,j}} \quad \text{for all } 1\leq j\leq m.\]
Then $\varphi\circ\tau_i=\tau_{q,z_i}\circ\varphi$ for each $i\in\{1,\ldots,r\}$, since for any $h\in\bK(\vz)$
\begin{align*} 
	\varphi(\tau_i(h(z_1,\ldots,z_m)))&= h(q^{b_{i,1}}u_1,\ldots,q^{b_{i,m}}u_m)\\
	&= \tau_{q,z_i}(\varphi(h(z_1,\ldots,z_m))).
\end{align*}
This directly implies the necessity in the second assertion. On the other hand, let $A=B^{-1}:=(a_{i,j})$ and we define a $\bK$-homomorphism $\tilde\varphi: \bK(\vz)\rightarrow\overline{\bK(\vz)}$ by $\tilde\varphi(\vz):=\vz^A$, i.e., 
\[w_j:=\tilde\varphi(z_j)=z_1^{a_{1,j}}\cdots z_m^{a_{m,j}}\quad \text{for all } 1\leq j\leq m.\]  Then  $\tilde\varphi\circ\tau_{q,z_i}=\tau_i\circ\tilde\varphi$ for each $i\in\{1,\ldots,r\}$, since for any $h\in\bK(\vz)$
\begin{align*}
	\tilde\varphi(\tau_{q,z_i}(h(z_1,\ldots,z_m)))&=h(w_1,\ldots,qw_i,\ldots,w_m)\\
	&=\tau_{i}(\tilde\varphi(h(z_1,\ldots,z_m))).
\end{align*}
The second equality follows from 
\[\tau_i(w_j)=\tau_{q,z_1}^{b_{i,1}}\cdots\tau_{q,z_m}^{b_{i,m}}(z_1^{a_{1,j}}\cdots z_m^{a_{m,j}})=q^{\sum_{\ell=1}^mb_{i,\ell}a_{\ell,j}}\cdot w_j=q^{\delta_{i,j}}\cdot w_j,\]
where $\delta_{i,j}$ is the Kronecker symbol. Moreover, $\tilde\varphi\circ\varphi=\one_{\bK(\vz)}$, since $\tilde\varphi( \varphi(\vz))=\tilde\varphi(\vz^{B})=\vz^{AB}=\vz$. Now suppose $\varphi(f)=\sum_{i=1}^r\Delta_{\tau_{q,z_i}}(h_i)$ for some $h_i\in\bK(\vz)$. Let $\tilde\varphi$ act on both sides of this equation. We get \[f=\Delta_{\tau_1}(\tilde h_1)+\cdots+\Delta_{\tau_r}(\tilde h_r),\]
where $\tilde h_i=\tilde\varphi(h_i)\in\overline{\bK(\vz)}$ for $1\leq i\leq r$. By Lemma~\ref{LEM:sumalgclosure}, $f$ is also $(\tau_1,\ldots,\tau_r)$-summable in $\bK(\vz)$.
\end{proof}

\begin{rem}\label{REM:mixdiffhome}
Let $f\in\bF(\vy,\vz)$ and $\{\si_1,\ldots,\si_{r_1}\}$, $\{\tau_1,\ldots,\tau_{r_2}\}$ be independent elements in $G^\si$ and $G^\tau$ respectively. Let $\tilde\si_i=c_i \si_i$ ($1\leq i \leq r_1$) and $\tilde\tau_j=\e_j \tau_j$ ($1\leq j\leq r_2$) for some $c_i,\e_j\in\bF^*$. We can define an $\F$-endomorphism $\psi\colon\bF(\vy,\vz)\to\bF(\vy,\vz)$ by 
\[\psi(\vy,\vz):=(\phi(\vy),\varphi(\vz)),\]
where $\phi,\varphi$ are given in Propositions \ref{PROP:diffhomo} and~\ref{PROP:qdiffhomo} respectively. Then $\psi \circ \si_i=\si_{y_{i}}\circ \psi$ for all $1\leq i \leq r_1$ and $\psi \circ \tau_{j}=\tau_{q,z_j}\circ \psi$ for all $1\leq j\leq r_2$. Furthermore, $f$ is $(\tilde\si_1,\ldots,\tilde\si_{r_1},\tilde\tau_1,\ldots,\tilde\tau_{r_2})$-summable in $\bF(\vy,\vz)$ if and only if $\psi(f)$ is $(\tilde\si_{y_1},\ldots,\tilde\si_{y_{r_1}},\tilde\tau_{q,z_1},\ldots,\tilde\tau_{q,z_{r_2}})$-summable in $\bF(\vy,\vz)$, where $\tilde\si_{y_i}=c_i\si_{y_i}$ and $\tilde\tau_{q,z_j}=\e_{j}\tau_{q,z_j}$.
\end{rem}
\begin{rem}
Note that in Proposition~\ref{PROP:qdiffhomo}, $\tilde{\varphi}(h_i)$ may not lie in~$\bK(\vz)$. We show that the algebraic extension can be avoided by choosing a proper~$\varphi$ when we apply it to solve Problem~\ref{PROB:summabilityproblem}. According to Remark~\ref{REM:mixeda}, we only need to consider the case in which $\{\tau_1,\ldots,\tau_{r_2}\}$ is a basis of~$G_d^\tau$. Note that $\tilde\varphi$ is a $\bK$-automorphism of~$\bK(\vz)$ if and only if $B$ is unimodular, i.e., the rows of $B$ form a $\bZ$-basis of~$\bZ^m$. The first $r_2$ rows of $B$ are determined by the operators~$\tau_1,\ldots,\tau_{r_2}$, as in Proposition~\ref{PROP:qdiffhomo}. To obtain a unimodular matrix~$B$, the remaining rows can be derived from $\tau_{r_2+1},\ldots,\tau_m \in G^\tau$ chosen such that their images~$ \bar{\tau}_{r_2+1},\ldots,\bar{\tau}_m$ form a basis of~$G^\tau/G_d^\tau$. The existence of such a choice is guaranteed by Lemma~\ref{LEM:G2/G2pfree} together with Lemma~7.4 in~\cite[Chapter~\uppercase\expandafter{\romannumeral3}]{BookLang}. 
\end{rem}
Combining Proposition~\ref{PROP:G/Gpfree},~Theorem~\ref{THM:rank=1} 
and Remark~\ref{REM:mixdiffhome}, we can reduce the summability Problem~\ref{PROB:summabilityproblem} in $n$ variables to that in fewer variables. 
With Lemma~\ref{LEM:unicase} serving as the base case, Problem~\ref{PROB:summabilityproblem} can be completely solved.
Furthermore, the $(\tau_1,\ldots,\tau_r)$-summability problem has been solved.

\begin{exam}
We determine whether $f$ is $( \si_{y_1},\si_{y_2},\tau_{q,z_1},\tau_{q,z_2} )$-summable over the rational function field~$\bQ(y_1,y_2,z_1,z_2)$, where 
\[
f := \frac{(z_1+z_2)/y_2}{z_1^2+(y_1-y_2)z_2^2}.
\]
Note that $f$ is of the form~\eqref{EQ:simplefraction} with $y_1$ playing the role of~$x_1$.
Let $a:=(z_1+z_2)/y_2,\, d := z_1^2+(y_1-y_2)z_2^2 $, and~$G := \la \si_{y_1},\si_{y_2},\tau_{q,z_1},\tau_{q,z_2} \ra$.
Then it can be verified that 
\[
G_d = \la  \si := \si_{y_1}\si_{y_2},\tau := \tau_{q,z_1}\tau_{q,z_2} \ra \ \text{with} \ \si(d) = d \ \text{and} \ \tau(d) = q^2d.
\]
By Remark~\ref{REM:mixeda}, we have $f$ is $( \si_{y_1},\si_{y_2},\tau_{q,z_1},\tau_{q,z_2} )$-summable if and only if $a = \Delta_{\si}(b)+\Delta_{q^{-2}\tau}(\lambda)$ for some $b, \lambda \in \bQ(\vy,\vz)$. According to Remark~\ref{REM:mixdiffhome}, we define a $\bQ$-automorphism $\psi\colon\bQ(\vy,\vz)\to\bQ(\vy,\vz)$ by 
\[\psi(y_1,y_2,z_1,z_2):=(y_1,y_1+y_2,z_1,z_1z_2).\]
Then $\psi^{-1}$ is given by 
\[\psi^{-1}(y_1,y_2,z_1,z_2)=(y_1,-y_1+y_2,z_1,z_2/z_1).\]
Now we consider the $(\si_{y_1},q^{-2}\tau_{q,z_1})$-summability of~$\psi(a) = \frac{(z_2+1)z_1}{y_1+y_2}$. Let $\tilde{d} := y_1+y_2 $ and~$\tilde G := \la \si_{y_1},\tau_{q,z_1} \ra$. Then we have $\tilde{G}_{\tilde{d}} = \la \tau_{q,z_1} \ra$ with~$\tau_{q,z_1}(\tilde d)= \tilde d$. 
Applying Equation~\eqref{EQ:q-sumpoly} leads to
\[
(z_2+1)z_1 = \Delta_{q^{-2}\tau_{q,z_1}}\biggl( \frac{z_2+1}{q^{-1}-1} z_1\biggl).
\]
Thus $\psi(a)$ is $( \si_{y_1},q^{-2}\tau_{q,z_1} )$-summable, with
\[
\psi(a) =\Delta_{q^{-2}\tau_{q,z_1}}\biggl(\frac{(z_2+1) z_1/(q^{-1}-1)}{y_1+y_2}\biggr).
\]
So $a$ is $( \si, q^{-2}\tau )$-summable, with
\[
a = \Delta_{q^{-2}\tau}\circ \psi^{-1} \biggl(\frac{(z_2+1) z_1}{(q^{-1}-1)(y_1+y_2)}\biggr) = \Delta_{q^{-2}\tau} \biggl( \frac{z_1+z_2}{(q^{-1}-1)y_2}
\biggr).
\]
Finally, we conclude that  
\begin{align*}
	f &= \frac{a}{d} = \frac{\Delta_{q^{-2}\tau}(b)}{d} = \tau_{q,z_1}\tau_{q,z_2}\Bigl( \frac{b}{d}\Bigr)-\frac{b}{d} \\
	&= \Delta_{\tau_{q,z_1}}\Bigl( \tau_{q,z_2}\Bigl(\frac{b}{d}\Bigr)\Bigr)+\Delta_{\tau_{q,z_2}}\Bigl(\frac{b}{d} \Bigr), 
	\ \text{where}~b := \frac{z_1+z_2}{(q^{-1}-1)y_2}.
\end{align*}
\end{exam}

\section{Examples for applications}\label{SEC:examples}
In this section, we give two examples to illustrate how to use the summability criteria to reduce multiple sums. It provides a potential tool for determining whether a series is convergent or irrational. Through this section, we assume that~$\bF = \bQ$.
\begin{exam}[Continuing Example~\ref{EX:qsum}]\label{EX:qsumreduce}
Identity~\eqref{EQ:examq} can be applied to reduce the following $n$-fold sum,
\[
\sum_{m_1,\ldots,m_n \geq 1} \frac{1}{q^{s\cdot m_1}+\cdots+q^{s\cdot m_n}} \ \text{with} \ |q|>1, \,s,n \in \bN^{+} \ \text{and} \ n >2.
\]
In order to translate the identity~\eqref{EQ:examq} into the usual sums, we define the transformation~$\varrho\colon\bQ(x_1,\ldots,x_n) \rightarrow \bQ(q^{m_1},\ldots,q^{m_n})$ by~$\varrho(x_i)=q^{m_i}$ for any~$i=1,\ldots,n$ and~$\varrho(c) = c$ for any~$c \in \bQ$. Since $q$ is not a root of unity, $\varrho$ is a $\bQ$-isomorphism between two fields, $\bQ(x_1,\ldots,x_n)$ and~$\bQ(q^{m_1},\ldots,q^{m_n})$. Let $\si_i$'s denote the shift operators with respect to $m_i$'s, for~$i=1,\ldots,n$. Then $\varrho(\tau_{q,x_i}(f))=\si_i(\varrho(f))$ for all~$f \in \bQ(x_1,\ldots,x_n)$. Now we derive the following identity:
\begin{align}\label{EQ:examq2}
	&\sum_{m_1,\ldots,m_n \geq 1} \frac{1}{q^{s\cdot m_1}+\cdots+q^{s\cdot m_n}} \nonumber\\
	=& \, \frac{1}{q^s-1}
	\biggl(\sum_{m_2,\ldots,m_n \geq 1} \frac{1}{1+q^{s\cdot m_2}+\cdots+q^{s\cdot m_n}} \nonumber\\
	&+\sum_{m_1,m_3\ldots,m_n \geq 1} \frac{1}{q^{s(m_1-1)}+1+q^{s\cdot m_3}+\cdots+q^{s\cdot m_n}} \nonumber\\
	&+ \cdots+ \sum_{m_1,\ldots,m_{n-1}\geq 1}\frac{1}{q^{s(m_1-1)}
		+\cdots+q^{s(m_{n-1}-1)}+1}\biggr). 
\end{align}
In this way, we reduce the $n$-fold sum to several $(n-1)$-fold sums.
Considering the first sum~$\sum_{m_2,\ldots,m_n \geq 1} \frac{1}{1+q^{s\cdot m_2}+\cdots+q^{s\cdot m_n}}$ on the right-hand side of Equation~\eqref{EQ:examq2}, we study the $(\tau_{q,x_2},\ldots, \tau_{q,x_n})$-summability of
\[\tilde f :=\frac{1}{1 + x_2^s+\cdots+x_n^s}.\]
Let $\tilde d := 1 + x_2^s+\cdots+x_n^s $ and~$\tilde G := \la \tau_{q,x_2},\ldots, \tau_{q,x_n} \ra$. It is obvious that~$\tilde G_{\tilde d} = \{ \operatorname{\one} \}$. And there does not exist $\tilde b \in \bQ(x_2,\ldots,x_n)$ such that~$ 1 = \one(\tilde b)- \tilde b $. Hence, $\tilde f$ is not $(\tau_{q,x_2},\ldots, \tau_{q,x_n})$-summable and we can not reduce the first sum into an $(n-2)$-fold sum. The similar issues arise when we consider the remaining sums in Equation~\eqref{EQ:examq2}.
\end{exam}
\begin{exam}
We show that $S(q)$ converges for~$0<|q|<1$, where 
\[
S(q) := \sum_{n,k,m\geq 1}\biggl(\frac{q^{k+m}}{[n+k+m]_q} - (1-q)q^{k+m}\biggr) \ \text{with} \ [N]_q:=\frac{1-q^N}{1-q}.
\]
There is a $\bQ$-isomorphism~$\varrho\colon \bQ(x,y,z) \rightarrow \bQ(q^n,q^k,q^m)$  defined analogously as in Example~\ref{EX:qsumreduce}. Specifically, $\varrho$ is given by  
\[\varrho(x,y,z) := (q^n, q^k, q^m).\]
First, we attempt to reduce the sum~$\sum_{n,k,m\geq 1}\frac{q^{k+m}}{[n+k+m]_q}$ by considering the $(\tau_{q,x},\tau_{q,y},\tau_{q,z})$-summability of~$(1-q)yz/(1-xyz)$. By calculation,
\[
\frac{(1-q)yz}{1-xyz}=\Delta_{\tau_{q,x}}\bigl(\frac{yz}{1-q^{-1}xyz}\bigr)-\Delta_{\tau_{q,y}}\bigl(\frac{yz}{1-q^{-1}xyz} \bigr).
\]
Applying this equation yields 
\[
\sum_{n,k,m\geq 1}\frac{q^{k+m}}{[n+k+m]_q} = \sum_{k,m\geq 1} \bigl(q^{k+m}-\frac{q^{k+m}}{1-q^{k+m}}\bigr)+\sum_{n,m\geq 1} \frac{q^{m+1}}{1-q^{n+m}}.
\]
It can be verified that for any~$|q|<1$, $\sum_{k,m \geq 1}\frac{q^{k+m}}{1-q^{k+m}}$ converges and
\[\sum_{k,m \geq 1}q^{k+m} = \bigl( \sum_{k\geq 1} q^k\bigr)^2=\bigl(\frac{q}{1-q}\bigr)^2.\]
Now we turn to the double sum~$\sum_{n,m\geq 1}\frac{q^{m+1}}{1-q^{n+m}}$. Considering the $(\tau_{q,x},\tau_{q,z})$-summability of~$qz/(1-xz)$, we have
\[
\frac{qz}{1-xz} = \Delta_{\tau_{q,x}}\bigl( \frac{qz}{(1-q)(1-q^{-1}xz)}\bigr)+\Delta_{\tau_{q,z}}\bigl( \frac{-qz}{(1-q)(1-q^{-1}xz)}\bigr).
\]
Hence, we have 
\begin{align*}
	\sum_{n,m\geq 1}\bigl( \frac{q^{m+1}}{1-q^{n+m}}\bigr) &= \frac{q}{1-q} \biggl(\sum_{m \geq 1} \bigl(q^m - \frac{q^m}{1-q^m}\bigr) +\sum_{n\geq 1} \frac{q}{1-q^n} \biggr)\\
	&=\bigl(\frac{q}{1-q}\bigr)^2 -q\sum_{m\geq 1}\frac{q^m}{1-q^m} + \sum_{n\geq 1}\frac{q^2}{1-q},
\end{align*}
where $\sum_{m \geq 1}\frac{q^m}{1-q^m}$ converges and $\sum_{n\geq 1}\frac{q^2}{1-q}$ diverges.
Applying the identity $\sum_{k,m \geq 1}q^{k+m}=\bigl(\frac{q}{1-q}\bigr)^2$ once more, we obtain that
\[
\sum_{n,k,m \geq 1} (1-q)q^{k+m} = \sum_{n\geq 1}\frac{q^2}{1-q}.
\]
In conclusion, we have~
\[
S(q) = 2\bigl(\frac{q}{1-q}\bigr)^2 - \sum_{k,m \geq 1}\frac{q^{k+m}}{1-q^{k+m}} -q\sum_{m\geq 1}\frac{q^m}{1-q^m},
\]
which is convergent.
Note that $\sum_{m\geq 1}\frac{q^m}{1-q^m}$ is a special case of Lambert series, see~\cite[p.\ 168]{BookJohnson}. The irrationality of the series~$\sum_{m\geq 1}\frac{q^m}{1-q^m}$ for $0<|q|<1$ was conjectured by Chowla~\cite{Chowla1947} and was proved by Erd\H{o}s~\cite{Erdos1948} when~$q=1/p \ \text{for} \  p \in \bZ\setminus \{0, \pm1\}$. 
Under the same assumption, the irrationality of $\sum_{k,m \geq 1}\frac{q^{k+m}}{1-q^{k+m}}$ follows from~\cite{Zudilin2005}, which in turn implies the irrationality of $S(q)$. 
\end{exam}
\section{Conclusion and future work}\label{SEC:conclusion}
The long-term project aims at developing algorithms for symbolic summation of multivariate functions.  This paper resolves the rational case and thus starts a testing step towards ($q$-)Gosper's algorithm for multivariate hypergeometric terms, which is formulated as in~\cite[Problem~5.1]{ChenKauers2017}. A large class of combinatorial identities involving single sums of hypergeometric terms can be efficiently proved by Zeilberger's method of creative telescoping~\cite{Zeilberger1990c,Zeilberger1991,PWZbook1996}, which is based on Gosper's algorithm. The multivariate extension of Gosper's algorithm will be crucial for improving the efficiency of creative telescoping when it is applied to prove identities involving  multiple sums with parameters. 

\bigskip

\noindent {\bf Acknowledgements.} 
We are grateful to the anonymous referees for their comments, which helped us significantly improve the paper.

\bibliographystyle{plain}










\end{document}